\documentclass[11pt]{article}
\usepackage{amsmath,amssymb,amsthm,epsfig}
\usepackage{graphicx}
\usepackage{multicol}
\usepackage{algorithmic,algorithm}
\usepackage{color}

\usepackage[left=1in,top=1in,right=1in,bottom=1in,nohead]{geometry}

\begin{document}
%Macros
\newcommand{\junk}[1]{}

\newtheorem{lemma}{Lemma}
\newtheorem{theorem}[lemma]{Theorem}
\newtheorem{informaltheorem}[lemma]{Informal Theorem}
\newtheorem{informallemma}[lemma]{Informal Lemma}
\newtheorem{corollary}[lemma]{Corollary}
\newtheorem{definition}[lemma]{Definition}
\newtheorem{proposition}[lemma]{Proposition}
\newtheorem{question}{Question}
\newtheorem{problem}{Problem}
\newtheorem{remark}[lemma]{Remark}
\newtheorem{claim}{Claim}
\newtheorem{fact}{Fact}
\newtheorem{challenge}{Challenge}
\newtheorem{observation}{Observation}
\newtheorem{openproblem}{Open Problem}
\newtheorem{openquestion}{Open question}
\newtheorem{conjecture}{Conjecture}
\newtheorem{game}{Game}

\newcommand{\beq}{\begin{equation}}
\newcommand{\eeq}{\end{equation}}
\newcommand{\beas}{\begin{eqnarray*}}
\newcommand{\eeas}{\end{eqnarray*}}

\newcommand{\poly}{\mathrm{poly}}
\newcommand{\eps}{\epsilon}
\newcommand{\e}{\epsilon}
\newcommand{\polylog}{\mathrm{polylog}}
\newcommand{\rob}[1]{\left( #1 \right)} %Round Brackets
\newcommand{\sqb}[1]{\left[ #1 \right]} %square Brackets
\newcommand{\cub}[1]{\left\{ #1 \right\} } %curly brackets
\newcommand{\rb}[1]{\left( #1 \right)} %Round
\newcommand{\abs}[1]{\left| #1 \right|} %| |
\newcommand{\zo}{\{0, 1\}}
\newcommand{\zonzo}{\zo^n \to \zo}
\newcommand{\zokzo}{\zo^k \to \zo}
\newcommand{\zot}{\{0,1,2\}}

\newcommand{\en}[1]{\marginpar{\textbf{#1}}}
\newcommand{\efn}[1]{\footnote{\textbf{#1}}}

\newcommand{\prob}[1]{\Pr \left[ #1 \right]}
\newcommand{\dprob}[3]{\Pr \left[ #1 ; #2 \mbox{;} #3 \right]}
\newcommand{\dsumprob}[2]{P(#1,#2)}
\newcommand{\expect}[1]{\mathbb{E}\left[ #1 \right]}
\newcommand{\mindegree}{\delta}

\newcommand{\BfPara}[1]{\noindent {\bf #1}.}

\newcommand{\cost}{\mbox{cost}}
\newcommand{\GNS}[1]{\mbox{{\sc GNS}}(#1)}
%End macros

\begin{titlepage}

\title{Information Spreading in Dynamic Networks}

\author{Chinmoy Dutta \thanks{College of Computer and Information
    Science, Northeastern University, Boston MA 02115, USA.  E-mail:
    {\tt \{chinmoy,rraj,austin\}@ccs.neu.edu}.  Chinmoy Dutta is
    supported in part by NSF grant CCF-0845003 and a Microsoft grant
    to Ravi Sundaram; Rajmohan Rajaraman and Zhifeng Sun are supported
    in part by NSF grant CNS-0915985.} \and Gopal Pandurangan
  \thanks{Division of Mathematical Sciences, Nanyang Technological
    University, Singapore 637371 and Department of Computer Science,
    Brown University, Providence, RI 02912, USA.  \hbox{E-mail}:~{\tt
      gopalpandurangan@gmail.com}. Supported in part by the following
    grants: Nanyang Technological University grant M58110000,
    Singapore Ministry of Education (MOE) Academic Research Fund
    (AcRF) Tier 2 grant MOE2010-T2-2-082, US NSF grants CCF-1023166
    and CNS-0915985, and a grant from the US-Israel Binational Science
    Foundation (BSF).}  \and Rajmohan Rajaraman~$^\ast$ \and Zhifeng
  Sun~$^\ast$}

\date{}

\maketitle

\thispagestyle{empty}

\begin{abstract}
We study the fundamental problem of information spreading (also known
as gossip) in dynamic networks.  In gossip, or more generally,
$k$-gossip, there are $k$ pieces of information (or tokens) that are
initially present in some nodes and the problem is to disseminate the
$k$ tokens to all nodes.  The goal is to accomplish the task in as few
rounds of distributed computation as possible.  The problem is
especially challenging in dynamic networks where the network topology
can change from round to round and can be controlled by an on-line
adversary.

The focus of this paper is on the power of token-forwarding
algorithms, which do not manipulate tokens in any way other than
storing and forwarding them.  We first consider a worst-case
adversarial model first studied by Kuhn, Lynch, and
Oshman~\cite{kuhn+lo:dynamic} in which the communication links for
each round are chosen by an adversary, and nodes do not know who their
neighbors for the current round are before they broadcast their
messages. \junk{The model allows the study of the fundamental
  computation power of dynamic networks.}Our main result is an
$\Omega(nk/\log n)$ lower bound on the number of rounds needed for any
deterministic token-forwarding algorithm to solve $k$-gossip.  This
resolves an open problem raised in~\cite{kuhn+lo:dynamic}, improving
their lower bound of $\Omega(n \log k)$, and matching their upper
bound of $O(nk)$ to within a logarithmic factor.  Our lower bound also
extends to randomized algorithms against an adversary that knows in
each round the outcomes of the random coin tosses in that round.  Our
result shows that one cannot obtain significantly efficient (i.e.,
subquadratic) token-forwarding algorithms for gossip in the
adversarial model of~\cite{kuhn+lo:dynamic}.  We next show that
token-forwarding algorithms can achieve subquadratic time in the
offline version of the problem, where the adversary has to commit all
the topology changes in advance at the beginning of the
computation. We present two polynomial-time offline token-forwarding
algorithms to solve $k$-gossip: (1) an $O(\min\{nk, n\sqrt{k \log
  n}\})$ round algorithm, and (2) an $(O(n^\eps), \log n)$ bicriteria
approximation algorithm, for any $\eps > 0$, which means that if $L$
is the number of rounds needed by an optimal algorithm, then our
approximation algorithm will complete in $O(n^\eps L)$ rounds and the
number of tokens transmitted on any edge is $O(\log n)$ in each
round. Our results are a step towards understanding the power and
limitation of token-forwarding algorithms in dynamic networks.
\end{abstract}

  
\end{titlepage}

\section{Introduction}
In a dynamic network, nodes (processors/end hosts) and communication
links can appear and disappear at will over time.  Emerging networking
technologies such as ad hoc wireless, sensor, and mobile networks,
overlay and peer-to-peer (P2P) networks are inherently dynamic,
resource-constrained, and unreliable.  This necessitates the
development of a solid theoretical foundation to design efficient,
robust, and scalable distributed algorithms and to understand the
power and limitations of distributed computing on such networks.  Such
a foundation is critical to realize the full potential of these
large-scale dynamic communication networks.

As a step towards understanding the fundamental computation power of
dynamic networks, we investigate dynamic networks in which the network
topology changes arbitrarily from round to round.  We first consider a
worst-case model that was introduced by Kuhn, Lynch, and
Oshman~\cite{kuhn+lo:dynamic} in which the communication links for
each round are chosen by an online adversary, and nodes do not know
who their neighbors for the current round are before they broadcast
their messages. (Note that in this model, only edges change and nodes
are assumed to be fixed.)  The only constraint on the adversary is
that the network should be connected at each round.  Unlike prior
models on dynamic networks, the model of~\cite{kuhn+lo:dynamic} does
not assume that the network eventually stops changing and requires
that the algorithms work correctly and terminate even in networks that
change continually over time.

In this paper, we study the fundamental problem of information
spreading (also known as gossip).  In gossip, or more generally,
$k$-gossip, there are $k$ pieces of information (or tokens) that are
initially present in some nodes and the problem is to disseminate the
$k$ tokens to all nodes.  (By just gossip, we mean $n$-gossip, where
$n$ is the network size.)  Information spreading is a fundamental
primitive in networks which can be used to solve other problem such as
broadcasting and leader election. Indeed, solving $n$-gossip, where
the number of tokens is equal to the number of nodes in the network,
and each node starts with exactly one token, allows any function of
the initial states of the nodes to be computed, assuming that the
nodes know $n$~\cite{kuhn+lo:dynamic}.  

\subsection{Our results}
The focus of this paper is on the power of {\em token-forwarding}\/
algorithms, which do not manipulate tokens in any way other than
storing and forwarding them.  Token-forwarding algorithms are simple,
often easy to implement, and typically incur low overhead.  In a key
result,~\cite{kuhn+lo:dynamic} showed that under their adversarial
model, $k$-gossip can be solved by token-forwarding in $O(nk)$ rounds,
but that any deterministic online token-forwarding algorithm needs
$\Omega(n \log k)$ rounds.  They also proved an $\Omega(nk)$ lower
bound for a special class of token-forwarding algorithms, called
knowledge-based algorithms.  Our main result is a new lower bound that
applies to {\em any}\/ deterministic online token-forwarding algorithm
for $k$-gossip.
\begin{itemize}
\item
We show that every online algorithm for the $k$-gossip problem takes
$\Omega(nk/\log n)$ rounds against an adversary that, at the start of
each round, knows the randomness used by the algorithm in the round.
This also implies that any deterministic online token-forwarding
algorithm takes $\Omega(nk/\log n)$ rounds.  Our result applies even
to centralized token-forwarding algorithms that have a global
knowledge of the token distribution.
\end{itemize}
This result resolves an open problem raised in~\cite{kuhn+lo:dynamic},
significantly improving their lower bound, and matching their upper
bound to within a logarithmic factor.  Our lower bound also enables a
better comparison of token-forwarding with an alternative approach
based on network coding due to
~\cite{haeupler:gossip,haeupler+k:dynamic}, which achieves a
$O(nk/\log n)$ rounds using $O(\log n)$-bit messages (which is not
significantly better than the $O(nk)$ bound using token-forwarding),
and $O(n + k)$ rounds with large message sizes (e.g., $\Theta(n \log
n)$ bits).  It thus follows that for large token and message sizes
there is a factor $\Omega(\min\{n,k\}/\log n)$ gap between
token-forwarding and network coding. We note that in our model we
allow only one token per edge per round and thus our bounds hold
regardless of the token size.

Our lower bound indicates that one cannot obtain efficient (i.e.,
subquadratic) token-forwarding algorithms for gossip in the
adversarial model of~\cite{kuhn+lo:dynamic}.  Furthermore, for
arbitrary token sizes, we do not know of any algorithm that is
significantly faster than quadratic time.  This motivates considering
other weaker (and perhaps, more realistic) models of dynamic networks.
In fact, it is not clear whether one can solve the problem
significantly faster even in an offline setting, in which the network
can change arbitrarily each round, but the entire evolution is known
to the algorithm in advance.  Our next contribution takes a step in
resolving this basic question for token-forwarding algorithms.
\begin{itemize}
\item
We present a polynomial-time offline token-forwarding algorithm that
solves the $k$-gossip problem on an $n$-node dynamic network in
$O(\min\{nk, n \sqrt{k \log n}\})$ rounds with high probability.
\item
We also present a polynomial-time offline token-forwarding algorithm
that solves the $k$-gossip problem in a number of rounds within an
$O(n^\eps)$ factor of the optimal, for any $\eps > 0$, assuming the
algorithm is allowed to transmit $O(\log n)$ tokens per round.
\end{itemize}
The above upper bounds show that in the offline setting,
token-forwarding algorithms can achieve a time bound that is within
$O(\sqrt{k\log n})$ of the information-theoretic lower bound of
$\Omega(n + k)$, and that we can approximate the best token-forwarding
algorithm to within a $O(n^\eps)$ factor, given logarithmic extra
bandwidth per edge.

\subsection{Related work}
Information spreading (or dissemination) in networks is one of the
most basic problems in computing and has a rich literature. \junk{Here
  we focus mostly on work that is relevant to our work.} The problem
is generally well-understood on static networks, both for
interconnection networks~\cite{leighton:book} as well as general
networks~\cite{lynch:distributed,attiya+w:distributed}.  In
particular, the $k$-gossip problem can be solved in $O(n + k)$ rounds
on any $n$-static network~\cite{topkis:disseminate}.  There also have
been several papers on broadcasting, multicasting, and related
problems in static heterogeneous and wireless networks (e.g.,
see~\cite{alon+blp:radio,bar-yehuda+gi:radio,bar-noy+gns:multicast,clementi+ms:radio}).

Dynamic networks have been studied extensively over the past three
decades.  Some of the early studies focused on dynamics that arise out
of faults, i.e., when edges or nodes fail.  A number of fault models,
varying according to extent and nature (e.g., probabilistic
vs.\ worst-case) and the resulting dynamic networks have been analyzed
(e.g., see~\cite{attiya+w:distributed,lynch:distributed}).  There have
been several studies on models that constrain the rate at which
changes occur, or assume that the network eventually stabilizes (e.g.,
see~\cite{afek+ag:dynamic,dolev:stabilize,gafni+b:link-reversal}).

There also has been considerable work on general dynamic networks.
Some of the earliest studies in this area
include~\cite{afek+gr:slide,awerbuch+pps:dynamic} which introduce
general building blocks for communication protocols on dynamic
networks.  \junk{Subsequently, a number of different problems have
  been studied on dynamic and asynchronous networks, including
  routing, load balancing, multicast, anycast, and several fundamental
  distributed computing problems.}  Another notable work is the local
balancing approach of~\cite{awerbuch+l:flow} for solving routing and
multicommodity flow problems on dynamic networks.  Algorithms based on
the local balancing approach continually balance the packet queues
across each edge of the network and drain packets that have reached
their destination.  \junk{It has been shown that assuming the queues
  at the nodes can hold enough packets, the local balancing approach
  can achieve throughput that is arbitrarily close to the optimal
  achievable by any offline algorithm.}  The local balancing approach
has been applied to achieve near-optimal throughput for multicast,
anycast, and broadcast problems on dynamic networks as well as for
mobile ad hoc
networks~\cite{awerbuch+bbs:route,awerbuch+bs:anycast,jia+rs:adhoc}.

Modeling general dynamic networks has gained renewed attention with
the recent advent of heterogeneous networks composed out of ad hoc,
and mobile devices.  To address the unpredictable and often unknown
nature of network dynamics,~\cite{kuhn+lo:dynamic} introduce a model
in which the communication graph can change completely from one round
to another, with the only constraint being that the network is
connected at each round.  The model of~\cite{kuhn+lo:dynamic} allows
for a much stronger adversary than the ones considered in past work on
general dynamic
networks~\cite{awerbuch+l:flow,awerbuch+bbs:route,awerbuch+bs:anycast}.
In addition to results on the $k$-gossip problem that we have
discussed earlier,~\cite{kuhn+lo:dynamic} consider the related problem
of counting, and generalize their results to the $T$-interval
connectivity model, which includes an additional constraint that any
interval of $T$ rounds has a stable connected spanning subgraph.  The
survey of~\cite{kuhn-survey} summarizes recent work on dynamic
networks.
 
We note that the model of~\cite{kuhn+lo:dynamic}, as well as ours,
allow only edge changes from round to round while the nodes remain
fixed. Recently, the work of \cite{p2p-soda} introduced a dynamic
network model (motivated by P2P networks) where both nodes and edges
can change by a large amount (up to a linear fraction of the network
size). They show that stable almost-everywhere agreement can be
efficiently solved in such networks even in adversarial dynamic
settings.  \junk{As in the Kuhn et al. model, the algorithms in
  \cite{p2p-soda} will work and terminate correctly even when the
  network keeps continually changing.  We note that there has been
considerable prior work in dynamic P2P networks (see \cite{p2p-soda,
  p2p-focs} and the references therein) but these don't assume that
the network keeps continually changing over time.}

Recent work of~\cite{haeupler:gossip,haeupler+k:dynamic} presents
information spreading algorithms based on network
coding~\cite{ahlswede+cly:coding}.  As mentioned earlier, one of their
important results is that the $k$-gossip problem on the adversarial
model of~\cite{kuhn+lo:dynamic} can be solved using network coding in
$O(n+k)$ rounds assuming the token sizes are sufficiently large
($\Omega(n\log n)$ bits). For further references to using network
coding for gossip and related problems, we refer to the recent works
of
~\cite{haeupler:gossip,haeupler+k:dynamic,avin1,avin2,deb+mc:coding,shah}
and the references therein.

Our offline approximation algorithm makes use of results on the
Steiner tree packing problem for directed
graphs~\cite{cheriyan+s:steiner}.  This problem is closely related to
the directed Steiner tree problem (a major open problem in
approximation
algorithms)~\cite{charikar+ccdgg:steiner,zosin+k:steiner} and the gap
between network coding and flow-based solutions for multicast in
arbitrary directed networks~\cite{agarwal+c:coding,sanders+et:flow}.

Finally, we note that there are also a number of studies that solve
$k$-gossip and related problems using {\em gossip-based}\/ processes.
In a local gossip-based algorithm, each node exchanges information
with a small number of randomly chosen neighbors in each round.
Gossip-based processes have recently received significant attention
because of their simplicity of implementation, scalability to large
network size, and their use in aggregate computations,
e.g.,~\cite{berenbrink+ceg:gossip,demers,kempe1,chen-spaa,karp,shah,boyd}
and the references therein.  All these studies assume an underlying
static communication network, and do not apply directly to the models
considered in this paper.  A related recent work on dynamic networks
is~\cite{avin+kl:dynamic} which analyzes the cover time of random
walks on dynamic networks.

\section{Model and problem statement}
\label{sec:pre}
In this section, we formally define the $k$-gossip problem, the online
and offline models, and token-forwarding algorithms.  

\smallskip
{\noindent {\sl The $k$-gossip problem.}} In this problem, $k$
different tokens are assigned to a set $V$ of $n \ge k$ nodes, where
each node may have any subset of the tokens, and the goal is to
disseminate all the $k$ tokens to all the nodes.

\smallskip
{\noindent {\sl The online model.}}  Our online model is the
worst-case adversarial model of~\cite{kuhn+lo:dynamic}.  Nodes
communicate with each other using anonymous broadcast.  We assume a
synchronized communication.  At the beginning of round $r$, each node
in $V$ decides what message to broadcast based on its internal state
and coin tosses (for a randomized algorithm); the adversary chooses
the set of edges that forms the communication network $G_r$ over $V$
for round $r$.  We adopt a {\em strong adversary}\/ model in which
adversary knows the outcomes of the random coin tosses used by the
algorithm in round $r$ at the time of constructing $G_r$ but is
unaware at this time of the outcomes of any randomness used by the
algorithm in future rounds.  The only constraint on $G_r$ is that it
be connected; this is the same as the $1$-interval connectivity model
of~\cite{kuhn+lo:dynamic}.  

As observed in~\cite{kuhn+lo:dynamic}, the above model is equivalent
to the adversary knowing the messages to be sent in round $r$ before
choosing the edges for round $r$.  We do not place any bound on the
size of the messages, but require for our lower bound that each
message contains at most one token.  Finally, we note that under the
strong adversary model, there is a distinction between randomized
algorithms and deterministic algorithms since a randomized algorithm
may be able to exploit the fact that in any round $r$, while the
adversary is aware of the randomness used in that round, it does not
know the outcomes of any randomness used in subsequent rounds.

\smallskip
{\noindent {\sl The offline model.}} In the offline model, we are
given a sequence of networks $\langle G_r \rangle$ where $G_r$ is a
connected communication network for round $r$.  As in the online
model, we assume that in each round at most one token is broadcast by
any node.  It can be easily seen that the $k$-gossip problem can be
solved in $nk$ rounds in the offline model; so we may assume that the
given sequence of networks is of length at most $nk$.

\smallskip
{\noindent {\sl Token-forwarding algorithms.}} Informally, a
token-forwarding algorithm is one that does not combine or alter
tokens, only stores and forwards them.  Formally, we call an algorithm
for $k$-gossip a token-forwarding algorithm if for every node $v$,
token $t$, and round $r$, $v$ contains $t$ at the start of round $r$
of the algorithm if and only if either $v$ has $t$ at the start of the
algorithm or $v$ received a message containing $t$ prior to round $r$.

Finally, several of our arguments are probabilistic.  We use the term
``with high probability'' to mean with probability at least $1 -
1/n^c$, for a constant $c$ that can be made sufficiently high by
adjusting related constant parameters.

\junk{
\begin{problem}[$k$-token dissemination]
\label{prob:ktoken}
\end{problem}

We assume synchronized communication between nodes, and the message
size is $O(\log n)$ where $n$ is the number of nodes in the graph. The
dynamic graphs are provided by adversaries. We consider the following
3 kinds of adversaries.
\begin{enumerate}
\item {\sc Strong Adversary}: At each round of communication, each
  node decides which token to send first. The adversary knows which
  node is going to send which token and the set of tokens each node
  has, and then the adversary provides a connected graph as the
  communication graph.
\item {\sc Weak Adversary}: At each round of communication, the
  adversary provides a connected graph as the communication graph
  first. Each node knows who are his neighbors, and then decides which
  token to send. The adversary knows the set of tokens each node has
  at any time.
\item {\sc Oblivious Adversary}: Before the token dissemination
  process starts, the adversary has to provide a sequence of graphs
  for all rounds of communications. 
\end{enumerate}
}

\section{Lower bound for online token-forwarding algorithms}
\label{sec:lower}
In this section, we give an $\Omega(kn/\log n)$ lower bound on the
number of rounds needed by any online token-forwarding algorithm for
the $k$-gossip problem against a strong adversary.  As discussed
earlier, this immediately implies the same lower bound for any
deterministic online token-forwarding algorithm.  Our lower bound
applies to even centralized algorithms and a large class of initial
token distributions.  We first describe the adversary strategy.

\smallskip
\noindent 
{\em Adversary:} The strategy of the adversary is simple.  We use the
notion of {\em free edge}\/ introduced in~\cite{kuhn+lo:dynamic}.  In
a given round $r$, we call an edge $(u,v)$ to be a free edge if at the
start of round $r$, $u$ has the token that $v$ broadcasts in the round
and $v$ has the token that $u$ broadcasts in the round\footnote{For
  convenience, when a node does not broadcast any token we will view
  it as broadcasting a special {\em empty}\/ token that every node
  has.  This allows us to avoid treating the empty broadcast as a
  special case.}; an edge that is not free is called {\em non-free}.
Thus, if $(u,v)$ is a free edge in a particular round, neither $u$ nor
$v$ can gain any new token through this edge in the round.  Since we
are considering a strong adversary model, at the start of each round,
the adversary knows for each node $v$, the token (if any) that $v$
will broadcast in that round.  In round $r$, the adversary constructs
the communication graph $G_r$ as follows.  First, the adversary adds
all the free edges to $G_r$.  Let $C_1,C_2,\dots,C_l$ denote the
connected components thus formed. The adversary then guarantees the
connectivity of the graph by selecting an arbitrary node in each
connected component and connecting them in a line.  Figure
\ref{fig:adversary} illustrates the construction.

The network $G_r$ thus constructed has exactly $l - 1$ non-free edges,
where $l$ is the number of connected components formed by the free
edges of $G_r$.  If $(u,v)$ is a non-free edge in $G_r$, then $u$,
$v$, or both will gain at most new token through this edge. We refer
to such a token exchange on a non-free edge as a {\em useful token
  exchange}.

We bound the running-time of any token-forwarding algorithm by
identifying a critical structure that quantifies the progress made in
each round.  We say that a sequence of nodes $v_1,v_2,\ldots,v_k$ is
{\em half-empty}\/ in round $r$ with respect to a sequence of tokens
$t_1,t_2,\ldots,t_k$ if the following condition holds at the start of
round $r$: for all $1 \le i,j \le k$, $i \neq j$, either $v_i$ is
missing $t_j$ or $v_j$ is missing $t_i$.  We then say that $\langle
v_i \rangle$ is half-empty with respect to $\langle t_i \rangle$ and
refer to the pair $(\langle v_i \rangle, \langle t_i \rangle)$ as a
half-empty configuration of size $k$.

\begin{figure}[ht]
\begin{center}
\includegraphics[width=5in]{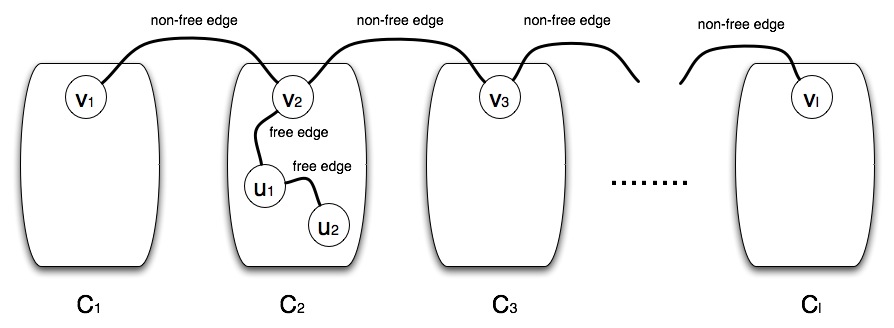}
\caption{The network constructed by the adversary in a particular
  round.  Note that if node $v_i$ broadcasts token $t_i$, then the
  $\langle v_i \rangle$ forms a half-empty configuration with respect
  to $\langle t_i \rangle$ at the start of this round.}
\label{fig:adversary}
\end{center}
\end{figure}

\junk{To prove the lower bound, we consider a special starting state, where
each node $u$ has token $t_i$ with probability $3/4$ for all $i$. Then
we look for certain structure in this starting state, and argue that
with such structure the number of useful token exchanges is $O(\log
n)$ with high probability. Furthermore, we argue that such structure
will always exist in the following rounds of communication. That
means, the number of useful token exchanges will be $O(\log n)$ in
every round. Since each node $u$ has any token with probability $3/4$
at the starting point, the expected number of missing token is $n\cdot
k/4 = \Omega(kn)$. In fact, we argue the number of missing token is
$\Omega(kn)$ with high probability using Chernoff bound. Thus, this
will imply any centralized deterministic algorithm runs in
$\Omega(kn/\log n)$ rounds.

First, we draw the connection between the number of useful token
exchanges and the existence of certain structure in Lemma
\ref{lem:free+edge}. Then in Lemma \ref{lem:alg+lower} we show the
$\Omega(kn/\log n)$ lower bound if the token dissemination process
starts from the state where each node $u$ has token $t_i$ with
probability $3/4$ for all $i$. Lastly, in Theorem
\ref{thm:alg+lower.single} we prove our lower bound.
\begin{lemma}
\label{lem:free+edge}
If $m$ or more useful token exchanges occur, then there exist $m/2+1$
nodes $v_1,v_2,\dots,v_{m/2+1}$ and $m/2+1$ tokens
$t_1,t_2,\dots,t_{m/2+1}$ ($v_i$ broadcasts token $t_i$) such that,
for all $i<j$, either $v_i$ is missing $t_j$ or $v_j$ is missing
$t_i$.
\end{lemma}
}

\begin{lemma}
\label{lem:free+edge}
If $m$ useful token exchanges occur in round $r$, then there exists a
half-empty configuration of size at least $m/2 + 1$ at the start of
round $r$.
\end{lemma}
\begin{proof}
Consider the network $G_r$ in round $r$.  Each non-free edge can
contribute at most 2 useful token exchanges. Thus, there are at least
$m/2$ non-free edges in the communication graph.  Based on the
adversary we consider, no useful token exchange takes place within the
connected components induced by the free edges. Useful token exchanges
can only happen over the non-free edges between connected
components. This implies there are at least $m/2+1$ connected
components in the subgraph of $G_r$ induced by the free edges.  Let
$v_i$ denote an arbitrary node in the $i$th connected component in
this subgraph, and let $t_i$ be the token broadcast by $v_i$ in round
$r$.  For $i \neq j$, since $v_i$ and $v_j$ are in different connected
components, $(v_i,v_j)$ is a non-free edge in round $r$; hence, at the
start of round $r$, either $v_i$ is missing $t_j$ or $v_j$ is missing
$t_i$.  Thus, the sequence $\langle v_i \rangle$ of nodes of size at
least $m/2 + 1$ is half-empty with respect to the sequence $\langle
t_i \rangle$ at the start of round $r$.
\end{proof}
An important point to note about the definition of a half-empty configuration
is that it only depends on the token distribution; it is independent
of the broadcast in any round.  This allows us to prove the following
easy lemma.
\begin{lemma}
\label{lem:monotone}
If a sequence $\langle v_i \rangle$ of nodes is half-empty with
respect to $\langle t_i \rangle$ at the start of round $r$, then
$\langle v_i \rangle$ is half-empty with respect to $\langle t_i
\rangle$ at the start of round $r'$ for any $r' \le r$.
\end{lemma}
\begin{proof}
The lemma follows immediately from the fact that if a node $v_i$ is
missing a token $t_j$ at the start of round $r$, then $v_i$ is missing
token $t_j$ at the start of every round $r' < r$.
\end{proof}
Lemmas~\ref{lem:free+edge} and~\ref{lem:monotone} suggest that if we
can identify a token distribution in which all half-empty
configuration are small, we can guarantee small progress in each
round.  We now show that there are many token distributions with this
property, thus yielding the desired lower bound.

\begin{theorem}
\label{thm:alg+lower}
From an initial token distribution in which each node has each token
independently with probability $3/4$, any online token-forwarding
algorithm will need $\Omega(kn/\log n)$ rounds to complete with high
probability against a strong adversary.
\end{theorem}

\begin{proof}
We first note that if the number of tokens $k$ is less than $100 \log
n$, then the $\Omega(kn/\log n)$ lower bound is trivially true because
even to disseminate one token it will take $\Omega(n)$ rounds in the
worst-case. Thus, in the following proof, we focus on the case where
$k \ge 100 \log n$.

Let $E_l$ denote the event that there exists a half-empty
configuration of size $l$ at the start of the first round.  For $E_l$
to hold, we need $l$ nodes $v_1, v_2, \dots, v_l$ and $l$ tokens
$t_1, t_2, \dots, t_l$ such that for all $i \neq j$ either $v_i$ is
missing $t_j$ or $v_j$ is missing $t_i$. For a pair of nodes $u$ and
$v$, by union bound, the probability that $u$ is missing $t_v$ or $v$
is missing $t_u$ is at most $1/4+1/4 = 1/2$. Thus, the probability of
$E_l$ can be bounded as follows.
\[ \prob{E_l} \le {n \choose l} \cdot \frac{k!}{(k-l)!} \cdot \rb{\frac{1}{2}}^{l \choose 2} \le n^l \cdot k^l \frac{1}{2^{l(l-1)/2}} \le \frac{2^{2l\log n}}{2^{l(l-1)/2}}. \]
In the above inequality, ${n \choose l}$ is the number of ways of
choosing the $l$ nodes that form the half-empty configuration,
$k!/(k-l)!$ is the number of ways of assigning $l$ distinct tokens,
and $(1/2)^{{l \choose 2}}$ is the upper bound on the probability for
each pair $i \neq j$ that either $v_i$ is missing $t_j$ or $v_j$ is
missing $t_i$.  For $l = 5 \log n$, $\prob{E_l} \le 1/n^2$.  Thus, the
largest half-empty configuration at the start of the first round, and
hence at the start of any round, is of size at most $5 \log n$ with
probability at least $1 - 1/n^2$.  By Lemma~\ref{lem:free+edge}, we
thus obtain that the number of useful token exchanges in each round is
at most $10 \log n$, with probability at least $1 - 1/n^2$.

\junk{Now we argue, in each following rounds, the number of useful token
exchanges is also no more than $2(l-1)$ with high probability, where $l\ge
5\log n$. If there are $2(l-1)$ or more useful token exchanges in
round $r$ where $r>1$, then by Lemma \ref{lem:free+edge} there exist
$l$ nodes $v_1,v_2,\dots,v_l$ and $l$ tokens $t_1,t_2,\dots,t_l$ such
that for all $i \neq j$ either $v_i$ is missing $t_j$ or $v_j$ is
missing $t_i$. Then at the beginning of round 1, the condition that
for all $i \neq j$ either $v_i$ is missing $t_j$ or $v_j$ is missing
$t_i$ still holds, and this can happen only with probability $1/n^2$.}

Let $M_i$ be the number of tokens that node $i$ is missing in the
initial distribution. Then $M_i$ is a binomial random variable with
$\expect{M_i} = k/4$.  By a straightforward Chernoff bound, we have the
probability that node $i$ misses less than $k/8$ tokens is
\[\prob{M_i \le \frac{k}{8}} = \prob{M_i \le \rb{1 - \frac{1}{2}} \cdot \expect{M_i}} \le e^{-\frac{\expect{M_i}\rb{\frac{1}{2}}^2}{2}} = e^{-\frac{k}{32}}. \] 
Therefore, the total number of tokens missing in the initial
distribution is at least $n \cdot k/8 = \Omega(kn)$ with probability
at least $1 - n/e^{\frac{k}{32}} \ge 1 - 1/n^2$ ($k \ge 100 \log n$).
Since the number of useful tokens exchanged in each round is at most
$10 \log n$, the number of rounds needed to complete $k$-gossip is
$\Omega(kn /\log n)$ with high probability.
\end{proof}

Theorem~\ref{thm:alg+lower} does not apply to certain natural initial
distributions, such as one in which each token resides at exactly one
node.  While this class of token distributions has far fewer tokens
distributed initially, the argument of Theorem~\ref{thm:alg+lower}
does not rule out the possibility that an algorithm, when starting
from a distribution in this class, avoids the problematic
configurations that arise in the proof. In the following,
Theorem~\ref{thm:lower.single} extends the lower bound to this class
of distributions.
\begin{lemma}
\label{lem:alg+lower.single}
From any distribution in which each token starts at exactly one node
and no node has more than one token, any online token-forwarding
algorithm for $k$-gossip needs $\Omega(kn/\log n)$ rounds against a
strong adversary.
\end{lemma}
\begin{proof}
We consider an initial distribution $C$ where each token is at exactly
one node, and no node has more than one token. Let $C^*$ be an initial
token distribution from which any online algorithm needs
$\Omega(kn/\log n)$ rounds.  The existence of $C^*$ follows from
Theorem~\ref{thm:alg+lower}.  We construct a bipartite graph on two
copies of $V$, $V_1$ and $V_2$. A node $v \in V_1$ is connected to a
node $u \in V_2$ if in $C^*$ $u$ has all the tokens that $v$ has in
$C$.  We will show below that this bipartite graph has a perfect
matching with positive probability.

Given a perfect matching $M$, we can complete the proof as follows.
For $v\in V_2$, let $M(v)$ denote the node in $V_1$ that got matched
to $v$.  If there is an algorithm $A$ that runs in $T$ rounds from
starting state $C$, then we can construct an algorithm $A^*$ that runs
in the same number of rounds from starting state $C^*$ as
follows. First every node $v$ deletes all its tokens except for those
which $M(v)$ has in $C$. Then algorithm $A^*$ runs exactly as $A$.
Thus, the lower bound of Theorem~\ref{thm:alg+lower}, which applies to
$A^*$, also applies to $A$.

It remains to prove that the above bipartite graph has a perfect
matching.  This follows from an application of Hall's Theorem.
Consider a set of $m$ nodes in $V_2$. We want to show their
neighborhood in the bipartite graph is of size at least $m$. We show
this condition holds by the following 2 cases. If $m < 3n/5$, let
$X_i$ denote the neighborhood size of node $i$. We know $\expect{X_i}
\ge 3n/4$. Then by Chernoff bound
\[\prob{X_i < m} \le \prob{X_i < 3n/5} \le e^{-\frac{\rb{1/5}^2 \expect{X_i}}{2}} = e^{-\frac{3n}{200}}\]
By union bound with probability at least $1-n\cdot e^{-3n/200}$ the
neighborhood size of every node is at least $m$. Therefore, the
condition holds in the first case. If $m \ge 3n/5$, we argue the
neighborhood size of any set of $m$ nodes is $V_1$ with high
probability. Consider a set of $m$ nodes, the probability that a given
token $t$ is missing in all these $m$ nodes is $(1/4)^m$. Thus the
probability that any token is missing in all these nodes is at most
$n(1/4)^m \le n(1/4)^{3n/5}$. There are at most $2^n$ such sets. By
union bound, with probability at least $1-2^n\cdot n(1/4)^{3n/5} =
1-n/2^{n/5}$, the condition holds in the second case.  
\end{proof}

\begin{theorem}
\label{thm:lower.single}
From any distribution in which each token starts at exactly one node,
any online token-forwarding algorithm for $k$-gossip needs
$\Omega(kn/\log n)$ rounds against a strong adversary.
\end{theorem}
\begin{proof}
In this theorem, we extend our proof in
Lemma~\ref{lem:alg+lower.single} to the inital distibution $C$ where
each token starts at exactly one node, but nodes may have multiple
tokens. We prove this theorem by the following two cases.

First case, when at least $n/2$ nodes start with some token. This
implies that $k\ge n/2$. Focus on the $n/2$ nodes with tokens. Each of
them has at least one unique token. By the same argument used in
Lemma~\ref{lem:alg+lower.single}, disseminating these $n/2$ distinct
tokens to $n$ nodes takes $\Omega(n^2/\log n)$ rounds. Thus, in this
case the number of rounds needed is $\Omega(kn/\log n)$.

Second case, when less than $n/2$ nodes start with some token. In this
case, the adversary can group these nodes together, and treat them as
one super node. There is only one edge connecting this super node to
the rest of the nodes. Thus, the number of useful token exchange
provided by this super node is at most one in each round. If there
exsits an algorithm that can disseminate $k$ tokens in $o(kn/\log n)$
rounds, then the contribution by the super node is $o(kn/\log n)$. And
by the same argument used in Lemma~\ref{lem:alg+lower.single} we know
dissemination $k$ tokens to $n/2$ nodes (those start with no tokens)
takes $\Omega(kn/\log n)$ rounds. Thus, the theorem also holds in this
case.
\end{proof}

\section{Subquadratic time offline token-forwarding algorithms}
\label{sec:centralized}
In this section, we give two centralized algorithms for the $k$-gossip
problem in the offline model. We present an $O(\min\{n^{1.5}\sqrt{\log
  n}, nk\})$ round algorithm in Section \ref{sec:upper}. Then we
present a bicriteria $\rb{O(n^\epsilon), \log n}$-approximation
algorithm in Section \ref{sec:approx}, which means if $L$ is the
number of rounds needed by an optimal algorithm where one token is
broadcast by every node per round, then our approximation algorithm
will complete in $O(n^\epsilon L)$ rounds and the number of tokens
broadcast by any node is $O(\log n)$ in any given round. Both of these
algorithms uses a directed capacitated leveled graph constructed from
the sequence of communication graphs which we call the {\em evolution
  graph}.

\smallskip
\noindent
{\em Evolution graph}: Let $V$ be the set of nodes. Consider a dynamic
network of $l$ rounds numbered $1$ through $l$ and let $G_i$ be the
communication graph for round $i$. The evolution graph for this
network is a directed capacitated graph $G$ with $2l+1$ levels
constructed as follows. We create $2l+1$ copies of $V$ and call them
$V_0, V_2, \dots, V_{2l}$. $V_i$ is the set of nodes at level $i$ and
for each node $v$ in $V$, we call its copy in $V_i$ as $v_i$. For $i =
1, \ldots, l$, level $2i-1$ corresponds to the beginning of round $i$
and level $2i$ corresponds to the end of round $i$. Level $0$
corresponds to the network at the start. Note that the end of a
particular round and the start of the next round are represented by
different levels. There are three kinds of edges in the graph. First,
for every round $i$ and every edge $(u,v) \in G_i$, we place two
directed edges with unit capacity each, one from $u_{2i-1}$ to
$v_{2i}$ and another from $v_{2i-1}$ to $u_{2i}$. We call these edges
{\em broadcast edges} as they will correspond to broadcasting of
tokens; the unit capacity on each such edge will ensure that only one
token can be sent from a node to a neighbor in one round. Second, for
every node $v$ in $V$ and every round $i$, we place an edge with
infinite capacity from $v_{2(i-1)}$ to $v_{2i}$. We call these edges
{\em buffer edges} as they ensure tokens can be stored at a node from
the end of one round to the end of the next. Finally, for every node
$v \in V$ and every round $i$, we also place an edge with unit
capacity from $v_{2(i-1)}$ to $v_{2i-1}$. We call these edges as {\em
  selection edges} as they correspond to every node selecting a token
out of those it has to broadcast in round $i$; the unit capacity
ensures that in a given round a node must send the same token to all
its neighbors. Figure \ref{fig:evolution} illustrates our
construction, and Lemma~\ref{lem:level.steiner} explains its
usefulness.

\begin{figure}[ht]
\begin{center}
\includegraphics[width=5in]{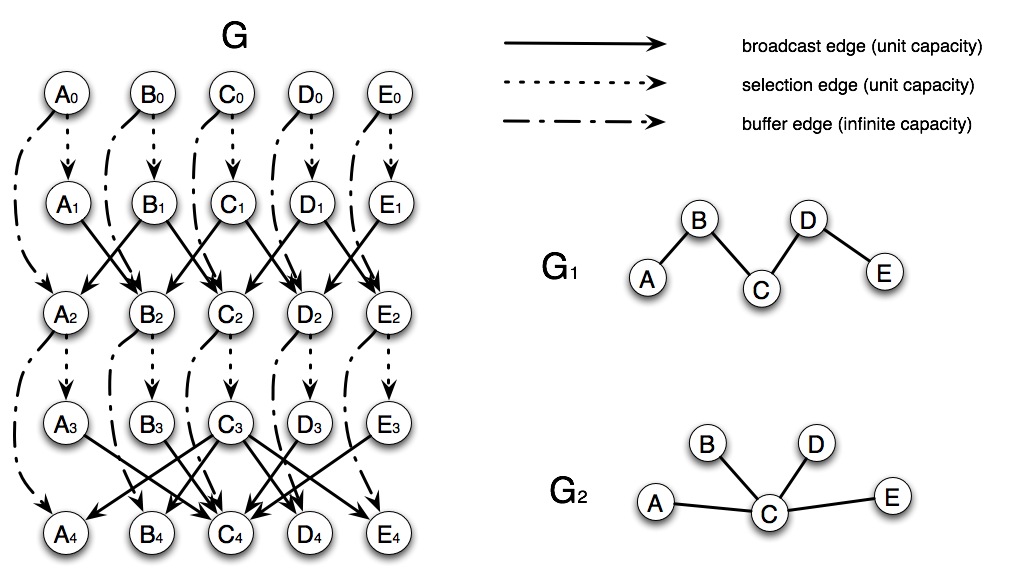}
\caption{An example of how to construct the evolution graph from a
  sequence of communication graphs.}
\label{fig:evolution}
\end{center}
\end{figure}

\begin{lemma}
\label{lem:level.steiner}
Let there be $k$ tokens, each with a source node where it is present
in the beginning and a set of destination nodes to whom we want to
send it. It is feasible to send all the tokens to all of their
destination nodes in a dynamic network using $l$ rounds, where in each
round a node can broadcast only one token to all its neighbors, if and
only if $k$ directed Steiner trees can be packed in the corresponding
evolution graph with $2l + 1$ levels respecting the edge capacities,
one for each token with its root being the copy of the source node at
level $0$ and its terminals being the copies of the destination nodes
at level $2l$.
\end{lemma}
\begin{proof}
Assume that $k$ tokens can be sent to all of their destinations in $l$
rounds and fix one broadcast schedule that achieves this. We will
construct $k$ directed Steiner trees as required by the lemma based on
how the tokens reach their destinations and then argue that they all
can be packed in the evolution graph respecting the edge
capacities. For a token $i$, we construct a Steiner tree $T^i$ as
follows.  For each level $j \in \{0, \ldots, 2l\}$, we define a set
$S^i_j$ of nodes at level $j$ inductively starting from level $2l$
backwards.  $S^i_{2l}$ is simply the copies of the destination nodes
for token $i$ at level $2l$. Once $S^i_{2(j+1)}$ is defined, we define
$S^i_{2j}$ (respectively $S^i_{2j+1}$) as: for each $v_{2(j+1)} \in
S^i_{2(j+1)}$, include $v_{2j}$ (respectively nothing) if token $i$
has reached node $v$ after round $j$, or include a node $u_{2j}$
(respectively $u_{2j+1}$) such that $u$ has token $i$ at the end of
round $j$ which it broadcasts in round $j+1$ and $(u,v)$ is an edge of
$G_{j+1}$. Such a node $u$ can always be found because whenever
$v_{2j}$ is included in $S^i_{2j}$, node $v$ has token $i$ by the end
of round $j$ which can be proved by backward induction staring from $j
= l$. It is easy to see that $S^i_0$ simply consists of the copy of
the source node of token $i$ at level $0$. $T^i$ is constructed on the
nodes in $\cup_{j = 0}^{j = 2l} S^i_j$. If for a vertex $v$,
$v_{2(j+1)} \in S^i_{2(j+1)}$ and $v_{2j} \in S^i_{2j}$, we add the
buffer edge $(v_{2j},v_{2(j+1)})$ in $T^i$. Otherwise, if $v_{2(j+1)}
\in S^i_{2(j+1)}$ but $v_{2j} \notin S^i_{2j}$, we add the selection
edge $(u_{2j},u_{2j+1})$ and broadcast edge $(u_{2j+1},v_{2(j+1)})$ in
$T^i$, where $u$ was the node chosen as described above. It is
straightforward to see that these edges form a directed Steiner tree
for token $i$ as required by the lemma which can be packed in the
evolution graph. The argument is completed by noting that any unit
capacity edge cannot be included in two different Steiner trees as we
started with a broadcast schedule where each node broadcasts a single
token to all its neighbors in one round, and thus all the $k$ Steiner
trees can be simultaneously packed in the evolution graph respecting
the edge capacities.

Next assume that $k$ Steiner trees as in the lemma can be packed in
the evolution graph respecting the edge capacities. We construct a
broadcast schedule for each token from its Steiner tree in the natural
way: whenever the Steiner tree $T_i$ corresponding to token $i$ uses a
broadcast edge $(u_{2j-1},v_{2j})$ for some $j$, we let the node $u$
broadcast token $i$ in round $j$. We need to show that this is a
feasible broadcast schedule. First we observe that two different
Steiner trees cannot use two broadcast edges starting from the same
node because every selection edge has unit capacity, thus there are no
conflicts in the schedule and each node is asked to broadcast at most
one token in each round. Next we claim by induction that if node
$v_{2j}$ is in $T^i$, then node $v$ has token $i$ by the end of round
$j$. For $j = 0$, it is trivial since only the copy of the source node
for token $i$ can be included in $T^i$ from level $0$. For $j > 0$, if
$v_{2j}$ is in $T^i$, we must reach there by following the buffer edge
$(v_{2(j-1)},v_{2j})$ or a broadcast edge $(u_{2j-1},v_{2j})$. In the
former case, by induction node $v$ has token $i$ after round $j-1$
itself. In the latter case, node $u$ which had token $i$ after round
$j-1$ by induction was the neighbor of node $v$ in $G_j$ and $u$
broadcast token $i$ in round $j$, thus implying node $v$ has token $i$
after round $j$. From the above claim, we conclude that whenever a
node is asked to broadcast a token in round $j$, it has the token by
the end of round $j-1$. Thus the schedule we constructed is a feasible
broadcast schedule. Since the copies of all the destination nodes of a
token at level $2l$ are the terminals of its Steiner tree, we conclude
all the tokens reach all of their destination nodes after round $l$.
\end{proof}

\begin{figure}[ht]
\begin{center}
\includegraphics[width=5in]{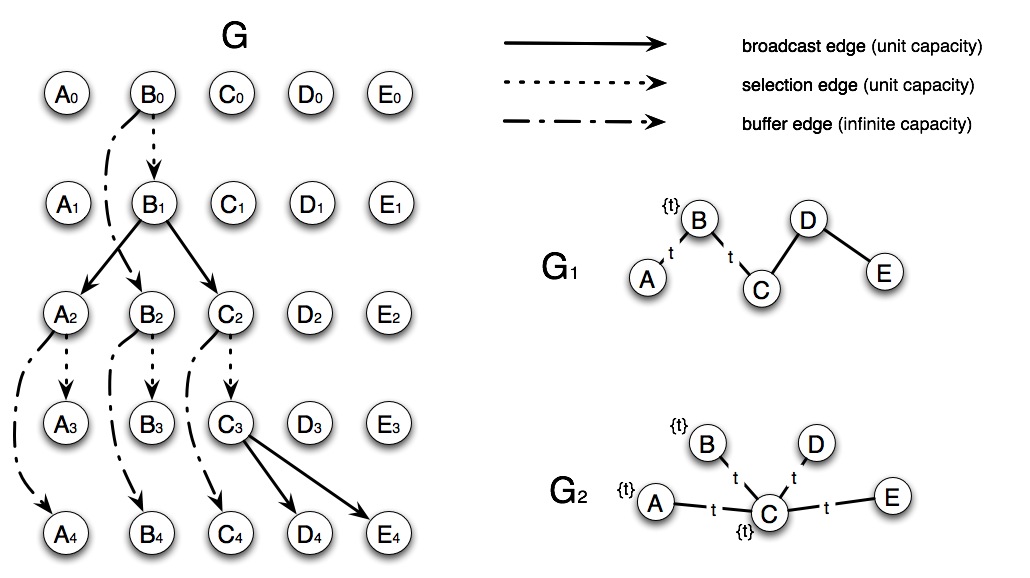}
\caption{An example of building directed Steiner tree in the evolution
  graph $G$ based on token dissemination process. Token $t$ starts
  from node $B$. Thus, the Steiner tree is rooted at $B_0$ in
  $G$. Since $B_0$ has token $t$, we include the infinite capacity
  buffer edge $(B_0,B_2)$. In the first round, node $B$ broadcasts
  token $t$, and hence we include the selection edge
  $(B_0,B_1)$. Nodes $A$ and $C$ receive token $t$ from $B$ in the
  first round, so we include edges $(B_1,A_2)$, $(B_1,C_2)$. Now
  $A_2$, $B_2$, and $C_2$ all have token $t$. Therefore we include the
  edges $(A_2,A_4)$, $(B_2,B_4)$, and $(C_2,C_4)$. In the second
  round, all of $A$, $B$, and $C$ broadcast token $t$, we include
  edges $(A_2,A_3)$, $(B_2,B_3)$, $(C_2,C_3)$. Nodes $D$ and $E$
  receive token $t$ from $C$. So we include edges $(C_3,D_4)$ and
  $(C_3,E_4)$. Notice that nodes $A$ and $B$ also receive token $t$
  from $C$, but they already have token $t$. Thus, we don't include
  edges $(C_3,B_4)$ or $(C_3,A_4)$.}
\label{fig:steiner}
\end{center}
\end{figure}

\subsection{An $O(\min\{n\sqrt{k\log n}, nk\})$ round algorithm}
\label{sec:upper}
Our algorithm is given in Algorithm~\ref{alg:flow_based} and analyzed
in Lemma~\ref{lem:level.flow} and~\ref{thm:flow_based}.

\begin{lemma}
\label{lem:level.flow}
Let there be $k \leq n$ tokens at given source nodes and let $v$ be an
arbitrary node. Then, all the tokens can be sent to $v$ using
broadcasts in $O(n)$ rounds.
\end{lemma}
\begin{proof}
By lemma~\ref{lem:level.steiner}, we will be done in $n + k$ rounds if
we can show that $k$ paths, one from every source vertex at level $0$
to $v_{2(n+k)}$, can be packed in the corresponding evolution graph
with $2(n+k) + 1$ levels respecting the edge capacities. For this, we
consider the evolution graph and add to it a special vertex $v_{-1}$
at level $-1$ and connect it to every source at level $0$ by an edge
of capacity 1. (Multiple edges get fused with corresponding increase
in capacity if multiple tokens have the same source.) We claim that
the value of the min-cut between $v_{-1}$ and $v_{2(n+k)}$ is at least
$k$. Before proving this, we complete the proof of the claim assuming
this.  By the max flow min cut theorem, the max flow between $v_{-1}$
and $v_{2(n+k)}$ is at least $k$. Since we connected $v_{-1}$ with
each of the $k$ token sources at level $0$ by a unit capacity edge, it
follows that unit flow can be routed from each of these sources at
level $0$ to $v_{2(n+k)}$ respecting the edge capacities. It is easy
to see that this implies we can pack $k$ paths, one from every source
vertex at level $0$ to $v_{2(n+k)}$, respecting the edge capacities.

To prove our claimed bound on the min cut, consider any cut of the
evolution graph separating $v_{-1}$ from $v_{2(n+k)}$ and let $S$ be
the set of the cut containing $v_{-1}$. If $S$ includes no vertex from
level $0$, we are immediately done. Otherwise, observe that if $v_{2j}
\in S$ for some $0 \leq j < (n+k)$ and $v_{2(j+1)} \notin S$, then the
value of the cut is infinite as it cuts the buffer edge of infinite
capacity out of $v_{2j}$. Thus we may assume that if $v_{2j} \in S$,
then $v_{2(j+1)} \in S$. Also observe that since each of the
communication graphs $G_1, \ldots, G_{n+k}$ are connected, if the
number of vertices in $S$ from level $2(j+1)$ is no more than the
number of vertices from level $2j$ and not all vertices from level
$2(j+1)$ are in $S$, we get at least a contribution of $1$ in the
value of the cut. But since the total number of nodes is $n$ and
$v_{2(n+k)} \notin S$, there must be at least $k$ such levels, which
proves the claim.
\end{proof}

\begin{algorithm}[ht!]
\caption{$O(\min\{n \sqrt{k\log n}, nk\})$ round algorithm in the
  offline model}
\label{alg:flow_based}
\begin{algorithmic}[1]
  \REQUIRE A sequence of communication graphs $G_i$, $i = 1, 2, \ldots$
  \ENSURE Schedule to disseminate $k$ tokens.

  \medskip

  \IF{$k \leq \sqrt{\log n}$}

  \FOR{each token $t$} \label{alg.step:flow_based.trivial}

  \STATE For the next $n$ rounds, let every node who has token
  $t$ broadcast the token.

  \ENDFOR 

  \ELSE

  \STATE Choose a set $S$ of $2\sqrt{k \log n}$ random nodes. \label{alg.step:random}
  
  \FOR{each vertex in $v \in S$} \label{alg.step:flow_based.phase_1}

  \STATE Send each of the $k$ tokens to vertex $v$ in $O(n)$ rounds. 

  \ENDFOR

  \FOR{each token $t$} \label{alg.step:flow_based.phase_2}

  \STATE For the next $2n \sqrt{(\log n)/k}$ rounds, let every node who has token
  $t$ broadcast the token.

  \ENDFOR

  \ENDIF

\end{algorithmic}
\end{algorithm}

\begin{theorem}
\label{thm:flow_based}
Algorithm~\ref{alg:flow_based} solves the $k$-gossip problem using
$O(\min\{n \sqrt{k \log n}, nk\})$ rounds with high probability in
the offline model.
\end{theorem}
\begin{proof}
It is trivial to see that if $k \leq \sqrt{\log n}$, then the
algorithm will end in $nk$ rounds and each node receives all the $k$
tokens. Assume $k > \sqrt{\log n}$. By Lemma~\ref{lem:level.flow}, all
the tokens can be sent to all the nodes in $S$ using $O(n \sqrt{k \log
  n})$ rounds. Now fix a node $v$ and a token $t$. Since token $t$ is
broadcast for $2n \sqrt{(\log n)/k}$ rounds, there is a set $S^t_v$ of
at least $2n \sqrt{(\log n)/k}$ nodes from which $v$ is reachable
within those rounds.  It is clear that if $S$ intersects $S^t_v$, $v$
will receive token $t$. Since the set $S$ was picked uniformly at
random, the probability that $S$ does not intersect $S^t_v$ is at most
\[ \frac{{n - 2n\sqrt{(\log n)/k} \choose 2\sqrt{k \log n}}}{{n \choose 2\sqrt{k \log n}}} < \left(\frac{n - 2n\sqrt{(\log n)/k}}{n}\right)^{2\sqrt{k \log n}} \le \frac{1}{n^4}. \] 
Thus every node receives every token with probability $1-1/n^3$. It is
also clear that the algorithm finishes in $O(n \sqrt{k \log n})$
rounds.
\end{proof}

Algorithm~\ref{alg:flow_based} can be derandomized using the standard
technique of conditional expectations, shown in
Algorithm~\ref{alg:derandomize}. Given a sequence of communication
graphs, if node $u$ broadcasts token $t$ for $\Delta$ rounds and every
node that receives token $t$ also broadcasts $t$ during that period,
then we say node $v$ is within $\Delta$ {\em broadcast distance} to
$u$ if and only if $v$ receives token $t$ by the end of round
$\Delta$. Let $S$ be a set of nodes, and $|S|\le 2 \sqrt{k \log
  n}$. We use $\dprob{u}{S}{T}$ to denote the probability that the
broadcast distance from node $u$ to set $X$ is greater than $2n
\sqrt{(\log n)/k}$, where $X=S\cup \cub{\mbox{pick }2\sqrt{k\log n} -
  |S|\mbox{ nodes uniformly at random from }V\setminus T}$, and
$\dsumprob{S}{T}$ denotes the sum, over all $u$ in $V$, of
$\dprob{u}{S}{T}$.

\begin{algorithm}[ht!]
\caption{Derandomized algorithm for Step~\ref{alg.step:random} in
  Algorithm~\ref{alg:flow_based}}
\label{alg:derandomize}
\begin{algorithmic}[1]
  \REQUIRE A sequence of communication graphs $G_i$, $i = 1, 2,
  \ldots$, and $k \ge \sqrt{\log n}$

  \ENSURE A set of $2\sqrt{k\log n}$ nodes $S$ such that the broadcast
  distance from every node $u$ to $S$ is within $2 n\sqrt{(\log
    n)/k}$.
  \medskip

  \STATE Set $S$ and $T$ be $\emptyset$.

  \FOR{each $v\in V$}

  \STATE $T = T \cup \{v\}$
  
  \IF{$\dsumprob{S \cup \{v\}}{T} \le \dsumprob{S}{T}$ \label{alg.step:cal}}

  \STATE $S = S\cup \{v\}$

%  \STATE Calculate $p=\dprob{V}{S \cup \{v\}}$. 
  
%  \IF{$p \ge 1- \frac{1}{n^3}$}

  \ENDIF

  \ENDFOR

  \RETURN $S$

\end{algorithmic}
\end{algorithm}

\begin{lemma}
The set $S$ returned by Algorithm~\ref{alg:derandomize} contains at
most $2\sqrt{k\log n}$ nodes, and the broadcast distance from every
node to $S$ is at most $2n\sqrt{(\log n)/k}$.
\end{lemma}
\begin{proof}
Let us view the process of randomly selecting $2\sqrt{k\log n}$ nodes
as a computation tree. This tree is a complete binary tree of height
$n$. There are $n+1$ nodes on any root-leaf path. The level of a node
is its distance from the root. The computation starts from the
root. Each node at the $i$th level is labeled by $b_i \in \{0,1\}$,
where 0 means not including node $i$ in the final set and 1 means
including node $i$ in the set. Thus, each root-leaf path, $b_1b_2\dots
b_n$, corresponds to a selection of nodes.  For a node $a$ in the
tree, let $S_a$ (resp., $T_a$) denote the sets of nodes that are
included (resp., lie) in the path from root to $a$.

By Theorem~\ref{thm:flow_based}, we know that for the root node $r$,
we have $\dsumprob{\emptyset}{S_r} =\dsumprob{\emptyset}{\emptyset}\le
1/n^3$.  If $c$ and $d$ are the children of $a$, then $T_c$ = $T_d$,
and there exists a real $0 \le p \le 1$ such that for each $u$ in $V$,
$\dprob{u}{S_a}{T_a}$ equals $p \dprob{u}{S_c}{T_c} +
(1-p)\dprob{u}{S_d}{T_d}$.  Therefore, $\dsumprob{S_a}{T_a}$ equals $p
\dsumprob{S_c}{T_c} + (1-p) \dsumprob{S_d}{T_d}$.  We thus obtain that
$\min\{\dsumprob{S_c}{T_c},\dsumprob{S_d}{T_d}\} \le
\dsumprob{S_a}{T_a}$.  Since we set $S$ to be $X$ in $\{S_c, S_d\}$
that minimizes $\dsumprob{X}{T_c}$, we maintain the invariant that
$\dsumprob{S}{T} \le 1/n^3$.  In particular, when the algorithm
reaches a leaf $l$, we know $\dsumprob{S_l}{V}\le 1/n^3$.  But a leaf
$l$ corresponds to a complete node selection, so that
$\dprob{u}{S_l}{V}$ is 0 or 1 for all $u$, and hence
$\dsumprob{S_l}{V}$ is an integer.  We thus have $\dsumprob{S_l}{V} =
0$, implying that the broadcast distance from node $u$ to set $S_l$ is
at most $2n \sqrt{(\log n)/k}$ for every $l$.  Furthermore, $|S_l|$ is
$2 k \sqrt{\log n}$ by construction.

Finally, note that Step~\ref{alg.step:cal} of
Algorithm~\ref{alg:derandomize} can be implemented in polynomial time,
since for each $u$ in $V$, $\dprob{u}{S}{T}$ is simply the ratio of
two binomial coefficients with a polynomial number of bits.  Thus,
Algorithm~\ref{alg:derandomize} is a polynomial time algorithm with
the desired property.
\end{proof}

\subsection{An $\rb{O(n^\epsilon), \log n}$-approximation algorithm}
\label{sec:approx}
Here we introduce an $\rb{O(n^\epsilon), \log n}$-approximation
algorithm for the $k$-gossip problem in the offline model. This means,
if the $k$-gossip problem can be solved on any $n$-node dynamic
network in $L$ rounds, then our algorithm will solve the $k$-gossip
problem on any dynamic network in $O(n^\epsilon L)$ rounds, assuming
each node is allowed to broadcast $O(\log n)$ tokens, instead of one,
in each round.  Our algorithm is an LP based one, which makes use of
the evolution graph defined earlier. The following is a
straightforward corollary of Lemma~\ref{lem:level.steiner}.

\begin{corollary}
\label{cor:level.steiner}
The $k$-gossip problem can be solved in $l$ rounds if $k$ directed
Steiner trees can be packed in the corresponding evolution graph,
where for each token, the root of its Steiner tree is a source node at
level 0, and the terminals are all the nodes at level $2l$.
\end{corollary}

Packing Steiner trees in general directed graphs is NP-hard to
approximate even within $\Omega(m^{1/3-\epsilon})$ for any
$\epsilon>0$ \cite{cheriyan+s:steiner}, where $m$ is the number of
edges in the graph. Thus, our algorithm focuses on solving Steiner
tree packing problem with relaxation on edge capacities, allowing the
capacity to blow up by a factor of $O(\log n)$. First, we write down
the LP for the Steiner tree packing problem (maximizing the number of
Steiner trees packed with respect to edge capacities). Let $\cal T$ be
the set of all possible Steiner trees, and $c_e$ be the capacity of
edge $e$. For each Steiner tree $T\in \cal T$, we associate a variable
$x_T$ with it. If $x_T=1$, then Steiner tree $T$ is in the optimal
solution; if $x_T=0$, it's not. After relaxing the integral
constraints on $x_T$'s, we have the following LP, referred to as $\cal
P$ henceforth. Let $F(\cal P)$ denote the optimal fractional solution
for $\cal P$.
\junk{
\begin{eqnarray*}
\max & \sum_{T\in \cal T} x_{T}\\
\mbox{s.t.}& \sum_{T:e\in T} x_{T} \le c_{e} \,\, \forall e\in E \\
 & x_{T} \ge 0 \,\, \forall T\in \cal T
\end{eqnarray*}
}
\[
\begin{array}{rrr}
\max & \sum_{T\in \cal T} x_{T} & \\
\mbox{s.t.} & \sum_{T:e\in T} x_{T} \le c_{e}  & \,\, \forall e\in E \\
 & x_{T} \ge 0  & \,\, \forall T\in \cal T
\end{array}
\]
\junk{
%%%%%%% begin junk %%%%%%%
\begin{eqnarray*}
\min & \sum_{e\in E} c_e y_e \\
\mbox{s.t.} & \sum_{e\in T} y_e \ge 1 \,\, \forall T\in \cal T \\
 & y_e \ge 0 \,\, \forall e\in E
\end{eqnarray*}
%%%%%%% end junk %%%%%%%
}

\junk{
%%%%%%% begin junk %%%%%%%
\begin{lemma}[\cite{jain+ms:steiner}]
\label{thm:approx-steiner}
There is an $\alpha$-approximation algorithm for the fractional
maximum Steiner tree packing problem if and only if there is an
$\alpha$-approximation algorithm for the minimum-weight Steiner tree
problem.
\end{lemma}
\textcolor{red}{Directed or undirected?}
Charikar et al. \cite{charikar+ccdgg:steiner} gives an
$O(n^\epsilon)$-approximation algorithm for the minimum-weight
directed Steiner tree problem. This together with Lemma
\ref{thm:approx-steiner} implies,
%%%%% end junk %%%%%%%
}

\begin{lemma}[\cite{cheriyan+s:steiner}]
\label{thm:approx-pack-steiner}
There is an $O(n^\epsilon)$-approximation algorithm for the fractional
maximum Steiner tree packing problem in directed graphs.
\end{lemma}

Let $L$ be the number of rounds that an optimal algorithm uses with
every node broadcasting at most one token per round. We give an
algorithm that takes $O(n^\epsilon L)$ rounds with every node
broadcasting $O(\log n)$ tokens per round. Thus ours is an
$\rb{O(n^\epsilon), O(\log n)}$ bicriteria approximation algorithm,
shown in Algorithm \ref{alg:approx}.
\begin{algorithm}[ht!]
\caption{$\rb{O(n^\epsilon), O(\log n)}$-approximation
  algorithm}
\label{alg:approx}
\begin{algorithmic}[1]
  \REQUIRE A sequence of communication graphs $G_1,G_2,\dots$
  \ENSURE Schedule to disseminate $k$ tokens.

  \medskip

  \STATE Initialize the set of Steiner trees ${\cal S} = \emptyset$.

  \FOR{$i = 1 \to 2n^\epsilon$}

  \STATE Find $L^*$ such that with the evolution graph $G$ constructed
  from level $0$ to level $2L^*$, the approximate value for $F(\cal
  P)$ is $k/n^{\epsilon}$. In this step, we use the algorithm of
  \cite{cheriyan+s:steiner} to approximate $F(\cal
  P)$. \label{alg.step:lp}

  \STATE Let $x^*_T$ be the value of the variable $x_T$ in the
  solution from step \ref{alg.step:lp}. The number of non-zero
  $x^*_T$'s is polynomial with respect to $k$. Using randomized
  rounding, with probability $x^*_T$ include $T$ in the solution,
  ${\cal S} = {\cal S} \cup \{T\}$. Otherwise, don't include
  $T$. \label{alg.step:round}

  \STATE Remove communication graphs $G_1,G_2,\dots,G_{L^*}$ from the
  sequence, and reduce the remaining graphs' indices by $L^*$.

  \ENDFOR

  \STATE Use Corollary \ref{cor:level.steiner} to convert the set of
  Steiner trees $\cal S$ into a token dissemination
  schedule. \label{alg.step:convert}
\end{algorithmic}
\end{algorithm}

\begin{theorem}
\label{thm:approx}
Algorithm \ref{alg:approx} achieves an $O(n^\epsilon)$ approximation
to the $k$-gossip problem while broadcasting $O(\log n)$ tokens per
round per node, with high probability.
\end{theorem}
\begin{proof}
We show the following three claims: (i) In Step
\ref{alg.step:convert}, $|{\cal S}| \ge k$ with probability at least
$1-1/e^{k/4}$. This is the correctness of Algorithm \ref{alg:approx},
saying it can find the schedule to disseminate all $k$ tokens. (ii)
The number of rounds in the schedule produced by Algorithm
\ref{alg:approx} is at most $O(n^\epsilon)$ times the optimal
one. (iii) In the token dissemination schedule, the number of tokens
sent over an edge is $O(\log n)$ in any round with high probability.

First, we prove claim (i). Let $X_i$ denote the sum of non-zero
$x^*_T$'s in iteration $i$. $X=\sum_{i=1}^{2n^\epsilon} X_i$. We know
$\expect{X_i} = k/n^{\epsilon}$. Thus, $\expect{X} = 2n^\epsilon
k/n^{\epsilon} = 2k$, which is the expected number of Steiner trees in
set $\cal S$. By Chernoff bound, we have 
\[\prob{X \le k} = \prob{X \le \rb{1-\frac{1}{2}}\expect{X}} \le e^{-\frac{\rb{1/2}^2 \expect{X}}{2}} = e^{-\frac{\rb{1/2}^2   \cdot 2k}{2}} = \frac{1}{e^{k/4}}\]
Thus, $|{\cal S}| \ge k$ with probability at least $1-1/e^{k/4}$ in
Step \ref{alg.step:convert}.

Next we prove claim (ii). Let $L$ denote the number of rounds needed
by an optimal algorithm. Since in Step \ref{alg.step:lp} we used the
$O(n^\epsilon)$-approximation algorithm in \cite{cheriyan+s:steiner}
to solve $F(\cal P)$, we know $L^* \le L$. There are $2n^\epsilon$
iterations. Thus, the number of rounds needed by Algorithm
\ref{alg:approx} is at most $2n^\epsilon L^* \le 2n^\epsilon L$, which
is an $O(n^\epsilon)$-approximation on the number of rounds.

Lastly we prove claim (iii). When Algorithm \ref{alg:approx} does
randomized rounding in Step \ref{alg.step:round}, some constraint
$\sum_{T:e\in T} x_{T} \le c_{e}$ in $\cal P$ may be violated. In the
evolution graph, $c_{e} = 1$. Let $Y$ denote the sum of $x^*_T$'s in
this constraint. We have $\expect{Y}\le c_e = 1$. By Chernoff bound,
\begin{eqnarray*}
\prob{Y\ge \expect{Y} + \log n} &=& \prob{Y \ge \rb{1+\frac{\log n}{\expect{Y}}} \expect{Y}} \\
 &\le& e^{-\expect{Y}\sqb{\rb{1+\frac{\log n}{\expect{Y}}} \ln \rb{1+\frac{\log n}{\expect{Y}}} - \frac{\log n}{\expect{Y}}}} \le \frac{1}{n^{\log\log n}} \\
\end{eqnarray*}
Thus, the number of tokens sent over a given edge is $O(\log n)$ with
probability at least $1-1/n^{\log\log n}$. Since there are only
polynomial number of edges, no edge will carry more than $O(\log n)$
tokens in a single round with high probability.
\end{proof}

%%%%%%%%%%%%%%%%%%%%%
%% \cite{lau:steiner}

\section{Conclusion and open questions}
In this paper, we studied the power of token-forwarding algorithms for
gossip in dynamic networks. We showed a lower bound of $\Omega(nk/\log
n)$ rounds for any online token forwarding algorithm against a strong
adversary; our bound matches the known upper bound of $O(nk)$ up to a
logarithmic factor.  This leaves us with an important open question:
what is the complexity of randomized online token-forwarding
algorithms against a weak adversary that is unaware of the randomness
used by the algorithm in each round?  We note that our lower bound
also extends to randomized algorithms if the adversary is allowed to
be adaptive; that is, the adversary is allowed to make its decision in
each step with knowledge of the random coin tosses made by the
algorithm in that step (but without knowledge of the randomness used
in future steps).  Furthermore, for small token sizes (e.g., $O(\log
n)$ bits) even the best (randomized) online algorithm we know based on
network coding takes $O(nk/\log n)$ rounds~\cite{haeupler+k:dynamic}.
In contrast, we show that in the offline setting there exist
centralized token-forwarding algorithms that run in
$O(n^{1.5}\sqrt{\log n})$ time.  An interesting open problem is to
obtain tight bounds on offline token-forwarding algorithms.

%\newpage

\section*{Acknowledgements}
\label{sec:ack}
We are grateful to Bernhard Haeupler and Fabian Kuhn for several
helpful discussions and comments on an earlier draft of the paper.  We
especially thank Bernhard for pointing out an improved analysis of
Algorithm~\ref{alg:flow_based} which yielded the current bound.

\bibliographystyle{plain}
\bibliography{refs}

\end{document}